\documentclass{article}
\usepackage{spconf,amsmath,graphicx}

\usepackage{cite}
\usepackage{amsmath,amssymb,amsfonts}

\usepackage{algorithmic}
\usepackage{graphicx}
\usepackage{textcomp}
\usepackage{xcolor}
\usepackage{tcolorbox}
\usepackage{booktabs} % To thicken table lines
\usepackage{amsthm} % Required for the proof environment
\usepackage{makecell, cellspace}
\usepackage{caption}
\usepackage{physics}
\usepackage{multirow}
\newtheorem{theorem}{Theorem}[section] % Theorem numbering is tied to sections
\newtheorem{proposition}[theorem]{Proposition} % Corollary numbering follows theorems
\def\BibTeX{{\rm B\kern-.05em{\sc i\kern-.025em b}\kern-.08em
    T\kern-.1667em\lower.7ex\hbox{E}\kern-.125emX}}

\newcommand{\bt}[1]{\mbox{$\bf #1$}}

\newcommand{\img}{\bt x}
\newcommand{\cimg}{\hat{\bt x}}

\newcommand{\nummb}{n_b}

\DeclareMathOperator*{\argmin}{arg\,min}
\DeclareMathOperator*{\expec}{\mathsf{E}}
\def\thetavec{\pmb{\theta}}

\title{Rate-Distortion Optimization with Non-Reference \\ Metrics for UGC compression 
}
\name{\hspace{-1.5em}Samuel Fernández-Menduiña$^*$, Xin Xiong$^*$, Eduardo Pavez$^*$,  Antonio Ortega$^*$, Neil Birkbeck$^\dagger$, Balu Adsumilli$^\dagger$ \thanks{This work was funded in part by a gift from YouTube.}}
\address{$^*$University of Southern California, Los Angeles, CA, USA\\ $^\dagger$Google Inc, Mountain View, CA, USA}

\begin{document}
\ninept
\maketitle

\begin{abstract}
Service providers must encode a large volume of noisy videos to meet the demand for user-generated content (UGC) in online video-sharing platforms.
However, low-quality UGC challenges conventional codecs based on rate-distortion optimization (RDO) with full-reference metrics (FRMs). 
While effective for pristine videos, FRMs drive codecs to preserve artifacts when the input is degraded, resulting in suboptimal compression. A more suitable approach used to assess UGC quality is based on non-reference metrics (NRMs). However, RDO with NRMs as a measure of distortion requires an iterative workflow of encoding, decoding, and metric evaluation, which is computationally impractical. 
This paper overcomes this limitation by linearizing the NRM around the uncompressed video. The resulting cost function enables block-wise bit allocation in the transform domain by estimating the alignment of the quantization error with the gradient of the NRM. 
To avoid large deviations from the input, we add sum of squared errors (SSE) regularization. We derive expressions for both the SSE regularization parameter and the Lagrangian, akin to the relationship used for SSE-RDO. Experiments with images and videos show bitrate savings of more than 30\% over SSE-RDO using the target NRM, with no decoder complexity overhead and minimal encoder complexity increase.
\end{abstract}

\section{Introduction}
Non-professional video, also known as user-generated content (UGC), plays a central role in platforms like YouTube and TikTok \cite{wang2019youtube}. 
UGC is often noisy due to amateur production, subpar equipment, and prior compression. Once uploaded, service providers usually re-encode these videos at different qualities and resolutions for streaming \cite{seufert2014survey}. 
However, low-quality UGC challenges conventional codecs based on rate-distortion optimization (RDO) \cite{sullivan_rate_1998} with full-reference metrics (FRMs) like the sum of squared errors (SSE). 
FRMs measure similarity to the input, converging to perfect quality as the compressed content becomes an identical copy of the original video. While this behavior is suitable for the compression of pristine content, it drives the codec to preserve artifacts in noisy UGC, leading to suboptimal compression. 

Non-reference metrics (NRMs), which assess visual quality without a reference, can be used instead of FRMs. NRMs are the preferred approach for evaluating compressed UGC videos \cite{pavez2022compression, wang2019youtube} because they reveal that increasing bitrate might not improve visual quality. 
For instance, Fig.~\ref{fig:quality_sat} shows that FRMs like the SSE and multi-scale structural similarity index measure (MS-SSIM) \cite{wang_multiscale_2003} approach zero distortion–i.e., perfect quality–at high bitrates, whereas NRMs, such as video score feature aggregation (VSFA), plateau at the input quality. 
This behavior shows the \emph{mismatch between UGC input quality and perfect quality} \cite{wang2019youtube}, which FRMs fail to capture.

While popular for UGC quality assessment, NRMs are seldom used for bit allocation. Since NRMs are computed on decoded videos, parameter selection with NRMs would require \emph{encoding, decoding, and metric evaluation for each coding option}, which is computationally impractical. 
Existing methods address this problem indirectly, constraining the coding parameter set by classifying the input based on its statistical properties or perceptual criteria and then applying a compression strategy to each class \cite{ling2020towards, john2020rate, yu2021predicting}. However, these methods require training on large datasets and parameters are preset for a class rather than optimized for each input. 
Alternatively, quality saturation detection methods \cite{pavez2022compression, xiong_rate_2023} identify the coding parameters at which investing more bits no longer improves quality. 
None of these techniques can help to adaptively select codec parameters block-wise based on the input to optimize a specific NRM.

We address this gap by developing a new approach that allows us to optimize coding for a given NRM via RDO. 
RDO based on SSE is efficient  \cite{li_asymptotic_1999} because 1) it is additive, so global parameter selection simplifies to block-level optimization, and 2) by Parseval's identity, we can compute SSE in the transform domain when orthogonal transforms are used. 
To obtain a cost function based on NRMs with these two properties, we approximate the computation of the NRM of the compressed video using Taylor's expansion (Fig.~\ref{fig:diagram_approx}) via autodiff \cite{paszke_automatic_2017}. 
This \emph{linearized NRM} (LNRM) depends on the gradient of the NRM with respect to the individual pixels in the input video, measuring the alignment of the reconstruction error with the direction of maximum variation of the NRM. 
Since the error in using Taylor's expansion is likely to increase as the rate decreases, we cannot guarantee a good approximation at all rates. 
To mitigate this effect, we add SSE regularization. We derive expressions for the SSE regularization parameter and for the Lagrangian, akin to existing SSE-RDO methods \cite{wiegand_lagrange_2001}. Since we only modify the RDO cost in the encoder, our compression method remains standard-compliant.

\begin{figure*}[t]
    \centering
    \includegraphics[width=\linewidth]{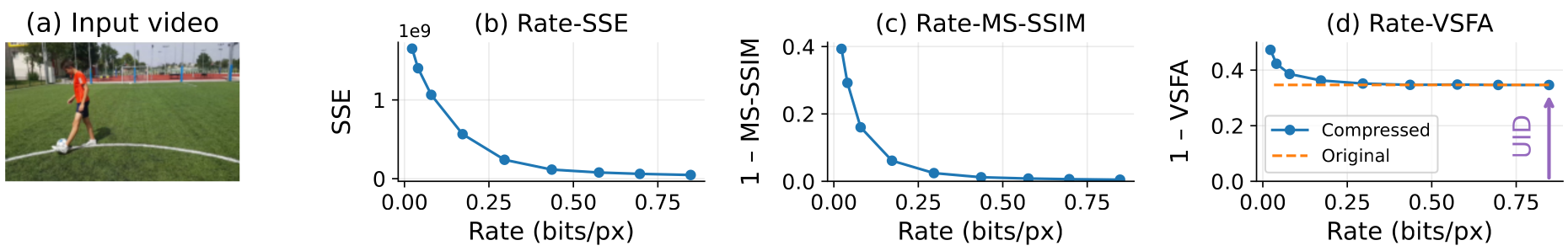}
    \caption{(a) Frame from a UGC video, and (b-d) rate-distortion curves obtained by compressing the input video with AVC. Unlike FRMs like SSE (b) and MS-SSIM (c), which converge to zero distortion as the rate increases, NRMs such as VSFA (d) do not converge to perfect quality for high bitrates, revealing the mismatch between zero distortion and the UGC input distortion (UID).}
    \label{fig:quality_sat}
\end{figure*}

\begin{figure}
    \centering
    \includegraphics[width=0.98\linewidth]{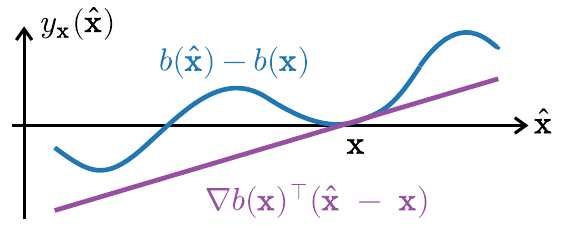}
    \caption{Gap $y_{\img}(\cimg)$ between the NRMs before and after compression, $b(\cimg) - b(\img)$, and the corresponding LNRM.}
    \label{fig:diagram_approx}
\end{figure}

We apply LNRM-RDO to AVC \cite{wiegand_overview_2003} with both images and videos using corresponding non-reference image \cite{mittal2012no, agnolucci2024arniqa} and video quality assessment methods \cite{li2019quality}. 
We test videos from the YouTube UGC dataset \cite{wang2019youtube}, setting the RDO to optimize different NRMs. Moreover, since our method can also be applied to pristine images, we consider pristine content from the KODAK dataset \cite{kodak1993kodak}. We also explore the trade-off introduced by the SSE regularizer. Results in Sec.~\ref{sec:exper} show that LNRM-RDO achieves bitrate savings of over $30\%$ when optimizing NRMs compared to conventional SSE-RDO, with no decoder complexity overhead and minimal encoder complexity increase. Given that we only modify the RDO, our approach can also be used for codecs with more flexible bit allocation schemes and more RDO options such as VVC \cite{bross2021overview} and AV1 \cite{han2021technical}.

\textbf{Notation.} Uppercase and lowercase bold letters, such as $\bt A$ and $\bt a$, denote matrices and vectors, respectively.
The $n$th entry of the vector $\bt a$ is $a_n$. Regular letters denote scalar values.

\section{Preliminaries}
\subsection{Rate-distortion optimization} 
Let $\img$ be the input image with $n_p$ pixels and $\cimg(\thetavec)$ its compressed version using parameters $\thetavec=[\theta_1 \ \theta_2 \ \hdots \theta_{n_b}]\in\Theta$, where $\Theta\in\mathbb{N}^{n_b}$ is the set of all possible operating points and $n_b$ the number of blocks in the image. Assume every entry of $\thetavec$ takes values in the set $\lbrace 1, \hdots, n_r\rbrace$, where $n_r$ denotes the number of RDO options. Given $\img_i$ for $i = 1, \hdots, \nummb$, we aim to find \cite{everett_generalized_1963}:
\begin{equation}
    \thetavec^\star = \argmin_{\thetavec \in \Theta} \, d(\img, \cimg(\thetavec)) + \lambda \, \sum_{i = 1}^{n_b}\, r_i(\cimg_i(\thetavec)),
\end{equation}
where $d(\cdot, \cdot)$ is the distortion metric, $r_i(\cdot)$ is the rate for the $i$th block, and $\lambda\geq 0$ is the Lagrange multiplier that controls the rate-distortion trade-off. We are especially interested in distortion metrics that are obtained as the sum of block-wise distortions,
\begin{equation}
\label{eq:local_gen}
d(\img_1, \hdots, \img_{\nummb}, \cimg_1(\thetavec), \hdots, \cimg_{\nummb}(\thetavec)) = \sum_{i = 1}^{\nummb} \, d_i(\img_i, \cimg_i(\thetavec)),
\end{equation}
which is true for SSE but may not hold for other metrics. Assuming that each block can be optimized independently \cite{ortega_rate-distortion_1998, sullivan_rate_1998}, we obtain $\cimg_i(\thetavec) = \cimg_i(\theta_i)$, which leads to 
\begin{equation}
\label{eq:final_form}
  \theta_i^\star = \argmin_{\theta_i \in \Theta_i} \, d_i(\img_i,  \cimg_i(\theta_i)) + \lambda \, r_{i}(\cimg_i(\theta_i)), \quad i = 1, \hdots, \nummb,  
\end{equation}
where $\Theta_i$ is the set of all possible parameters for the $i$th block. This is the RDO formulation most video codecs solve \cite{fernandez2024feature}. Given a quality parameter $\mathrm{QP}$, a practical way to control the rate-distortion trade-off is by setting \cite{wiegand_lagrange_2001}:
\begin{equation}
\label{eq:og_multiplier}
\lambda = c \, 2^{(\mathrm{QP}-12) / 3},
\end{equation}
where $c$ varies with the type of frame and content \cite{ringis_disparity_2023}. This work aims to replicate this block-level RDO formulation using distortion metrics derived from non-reference metrics.

\begin{figure*}[t]
    \centering
    \includegraphics[width=0.98\linewidth]{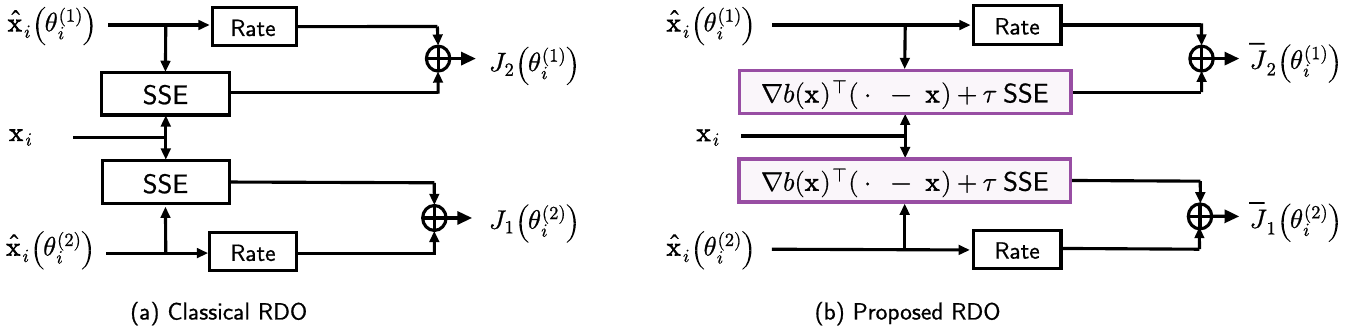}
    \caption{RDO with two options, $\theta_i^{(1)}$ and $\theta_i^{(2)}$, using (a) classical SSE-RDO and (b) LNRM-RDO. LNRM-RDO computes the distortion using the gradient of the NRM, which encodes information about the quality of the video.}
    \label{fig:rdo_blocks}
\end{figure*}
\subsection{Non-reference metrics}
Non-reference metrics (NRMs) measure image or video quality without relying on a reference, assessing distortions directly from the content, which makes them suitable for UGC applications. Although our techniques apply to any NRM, we focus on popular methods that have been trained on compressed data. For images, we use 1) the blind/referenceless image spatial quality evaluator (BRISQUE) \cite{mittal2012no}, which leverages natural scene statistics to detect unnatural deviations caused by distortions, and 2) leArning distoRtion maNifold for Image Quality Assessment (ARNIQA) \cite{agnolucci2024arniqa} trained on TID2013 \cite{ponomarenko2015image}, which achieves state-of-the-art results by modeling distortion manifolds via deep learning. For videos, we consider the Video Score Feature Aggregation (VSFA) \cite{li2019quality}. VSFA incorporates spatial and temporal features to model perceptual quality dynamically using a convolutional neural network. For BRISQUE, a lower value means better quality (↓), while for ARNIQA and VSFA, a higher value means better quality (↑).

\section{Towards RDO with Linearized NRMs}
\subsection{RDO with NRMs}
\label{sec:taylor} 
We denote a given NRM as $b(\cdot)$. In general, $b(\cdot)$ will be a computational metric. We assume lower  values mean higher quality. Direct bit allocation to optimize an NRM leads to:
\begin{equation}
    \label{eq:direct-RDO}
    \thetavec^\star = \argmin_{\thetavec\in\Theta}\ b(\cimg(\thetavec)) - b(\img) + \lambda \, \sum_{i = 1}^{n_b}\, r_i(\cimg_i(\thetavec)).
\end{equation}
Unlike SSE, the proposed distortion term $b(\cimg(\thetavec)) - b(\img)$ can be negative, which indicates that the reconstructed image has better quality (as defined by the NRM) than the original input. When $b(\cimg(\thetavec)) \geq b(\img)$, we reward (rather than penalize) that particular RDO option. This aligns with the observation that, since UGC videos are already noisy, the quality of the compressed content may surpass the quality of the input UGC \cite{wang2019youtube}. 
We highlight that resource allocation via Lagrangian relaxation allows for negative cost functions \cite{everett_generalized_1963}. 

Finding the optimal solution to \eqref{eq:direct-RDO} would require obtaining the complete decoded video in the pixel domain to compute the NRM. 
Hence, parameter selection would require encoding the whole video with each possible coding option and assessing the final result, which is impractical. We aim to modify \eqref{eq:direct-RDO} so that the RDO can be made block-wise during encoding.

\subsection{Linearization}
Define the gradient of the NRM with respect to the pixels of the input image as:
\begin{equation}
    \nabla b(\img) \doteq \begin{bmatrix} \displaystyle \frac{\partial b(\img)}{\partial x_1} & \displaystyle \frac{\partial b(\img)}{\partial x_2} & \hdots & \displaystyle \frac{\partial b(\img)}{\partial x_{n_p}}  \end{bmatrix}^\top.
\end{equation}
We linearize $b(\cimg)$ around the uncompressed input (cf.~Fig.~\ref{fig:diagram_approx}):
\begin{equation}
    b(\cimg(\thetavec)) = b(\img) + \grad b(\img)^\top (\cimg(\thetavec) - \img) + o(\norm{\cimg(\thetavec) - \img}_2^2),
\end{equation}
where $o(x)$ converges to zero as fast as $x$. Hence,
\begin{equation}
b(\cimg(\thetavec)) - b(\img) \cong \grad b(\img)^\top (\cimg(\thetavec) - \img),
\end{equation}
where $\cong$ denotes convergence for high bitrates  \cite{linder_high-resolution_1999}. We call this loss linearized NRM (LNRM), which measures the alignment of the reconstruction error with the gradient of the metric, encouraging or penalizing error patterns that improve or decrease the NRM of the video, respectively. Thus, the RDO problem becomes:
\begin{equation}
    \thetavec^\star = \argmin_{\thetavec\in\Theta}\ \grad b(\img)^\top (\cimg(\thetavec) - \img) + \lambda \, \sum_{i=1}^{n_b} r_i(\cimg_i(\thetavec)).
    \label{eq:optimization-1}
\end{equation}
This cost is already amenable to block-level evaluation. Let $\grad b_i(\img)$ be the gradient entries corresponding to the $i$th block. Then,
\begin{equation}
    \label{eq:distortion}
    \grad b(\img)^\top (\cimg(\thetavec) - \img) = \sum_{i = 1}^{n_b} \, \grad b_i(\img)^\top(\cimg_i(\thetavec) - \img_i),
\end{equation}
Assuming $\cimg_i(\thetavec) = \cimg_i(\theta_i)$, we can re-write the RDO problem as
\begin{equation}
\label{eq:pix_domain}
    \theta^\star_i = \argmin_{\theta_i\in\Theta_i} \ \grad b_i(\img)^\top(\cimg_i(\theta_i) - \img_i) + \lambda\, r_i(\cimg_i(\theta_i)),
\end{equation}
for $i = 1, \hdots, n_b$. Applying this cost function introduces some challenges, which we address by using SSE regularization.

\subsection{SSE regularization}
\label{sec:multiplier}
Since Taylor's expansion is more accurate near the input, LNRM is most reliable in the high bit-rate regime. However, as bit-rate decreases, relying completely on the direction of the gradient may lead to spurious results \cite{liu2024defense}. Moreover, the derivation of the Lagrange multiplier in \eqref{eq:og_multiplier}, which is a staple of modern video codecs \cite{bross2021overview}, is based on the quadratic relationship between the quantization step and the expectation of the SSE under the uniform noise model \cite{sullivan_rate_1998}. However, LNRM yields an expected value of zero under the uniform model, regardless of the quantization level. While we can find  $\lambda$ for LNRM using line search, this process is time-consuming. 
We address these problems by adding a constraint to \eqref{eq:optimization-1}: 
\begin{multline}
\thetavec^\star = \argmin_{\thetavec\in\Theta} \ \nabla b(\img)^\top(\cimg(\thetavec) - \img) + \lambda \, \sum_{i = 1}^{n_b}r_i(\cimg_i(\theta_i)), \\
\mathrm{such \  that} \ \norm{\cimg(\thetavec) - \img}_2^2 \leq \mathrm{SSE}_{\mathrm{max}},
\end{multline}
where $\mathrm{SSE}_{\mathrm{max}}$ is the maximum acceptable SSE with respect to the input video.
Using Lagrangian relaxation \cite{everett_generalized_1963} and grouping together terms other than the rate, we define:
\begin{equation}
\label{eq:sse_reg}
    d(\img, \cimg(\thetavec)) \doteq \nabla b(\img)^\top (\cimg(\thetavec) - \img) + \tau \norm{\cimg(\thetavec) - \img}_2^2,
\end{equation}
where $\tau$ is a parameter that can be selected by the user to control the trade-off between SSE and NRM performance. We call this optimization process LNRM-RDO; with some abuse of notation, we will particularize to the specific metric (e.g., BRISQUE-RDO) when needed. We compare LNRM-RDO to SSE-RDO in Fig.~\ref{fig:rdo_blocks}.

\begin{figure*}[t]
    \centering
    \includegraphics[width=0.49\linewidth]{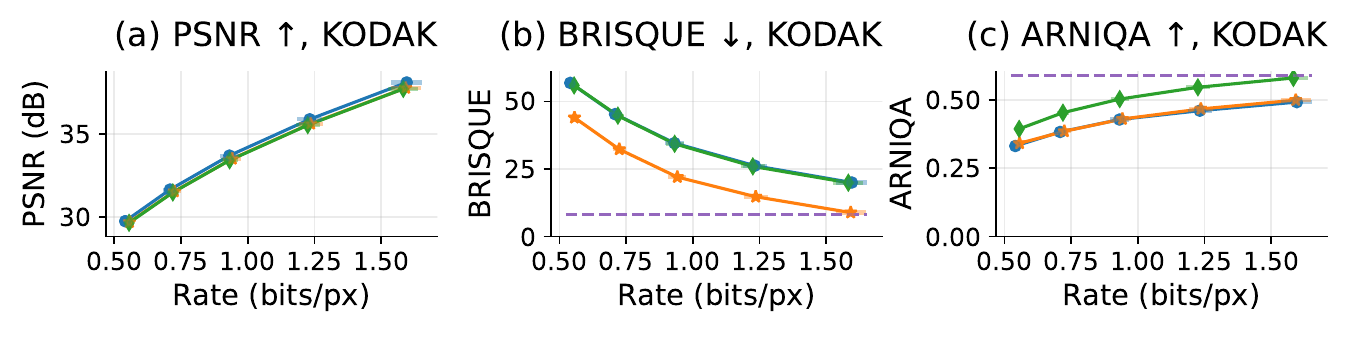}
    \includegraphics[width=0.49\linewidth]{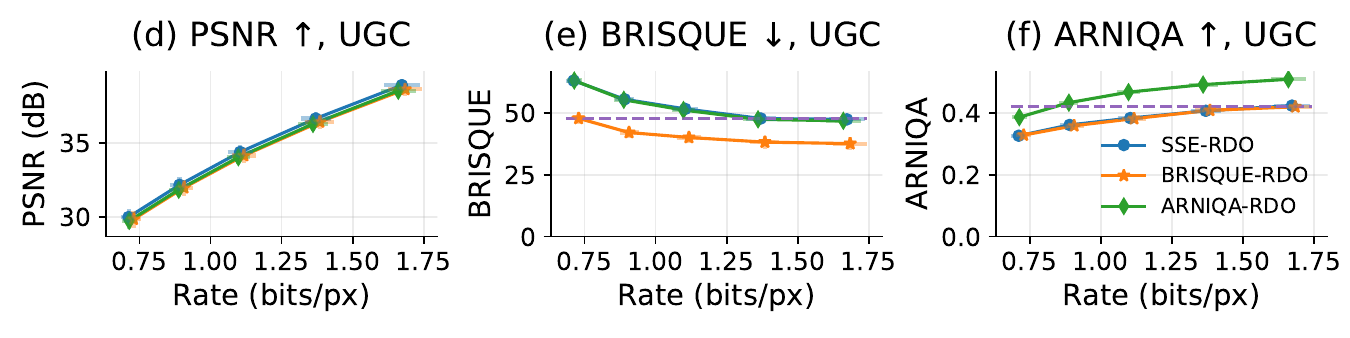}
    \caption{Average RD curves, with standard error for both axes, for KODAK (a-c) and selected frames from the YouTube-UGC dataset (d-f) using different RDO methods. For NRMs, the average value of the input images is shown with a dashed purple line. The best results for each metric are obtained by the RDO targeting that metric. For UGC, the NRM of the compressed image may surpass the NRM of the input \cite{wang2019youtube}.}
    \label{fig:rd_pristine}
    \vspace{-1em}
\end{figure*}

To derive an expression for the Lagrange multiplier $\lambda$ using rate-distortion theory \cite{berger2003rate}, the expected distortion is often modeled as an exponential function of the rate \cite{linder_high-resolution_1999}, which implies that it must be positive. Hence, we must choose $\tau$ so that the expectation of the term in \eqref{eq:sse_reg} is positive. Under the uniform quantization model \cite{gish1968asymptotically}, any $\tau > 0$ satisfies this condition. 

Additionally, since different NRMs may have different scales, the magnitude of the gradient $\nabla b(\img)$ may change based on the metric. Thus, the impact of the LNRM term in the regularized distortion will differ depending on the target NRM. To account for this variability and avoid exhaustive search loops, we aim to choose a different $\tau$ for each LNRM. We split $\tau$ into a normalized factor $\tilde{\tau}$ and a scaling term $\alpha$. Even though in the final optimization the relative importance of SSE and the NRM is controlled by modifying $\alpha$, it is useful to start by selecting a normalized $\tilde{\tau}$ to simplify the selection of $\alpha$. 
As a heuristic, we choose $\tilde{\tau}$ so that, for the quantization errors that yield the maximum possible values of the LNRM and the SSE, i.e., for the worst-case distortion measured by each metric,  both terms in \eqref{eq:sse_reg} have equal contribution to the final loss.
\begin{proposition}
    Under uniform quantization with step $\Delta$,
    \begin{equation}
    \label{eq:tau}
        \tilde{\tau} = 2 / \sqrt{n_p} \, \norm{\nabla b(\img)}_2 / \Delta,
    \end{equation}
    ensures that $\max_{\thetavec} \, \nabla b(\img)^\top (\cimg(\thetavec) - \img) = \tilde{\tau} \max_{\thetavec}  \norm{\cimg(\thetavec) - \img}_2^2$.
\end{proposition}
\begin{proof}
We aim to find $$\tilde{\tau} = \max_{\thetavec} \, \nabla b(\img)^\top (\cimg(\thetavec) - \img) / \max_{\thetavec}  \norm{\cimg(\thetavec) - \img}_2^2.$$
By Cauchy-Schwarz, the maximum of the NRM happens when the direction of the error aligns with the gradient, $
\nabla b(\img)^\top (\cimg(\thetavec) - \img) \leq  \norm{\nabla b(\img)}_2 \norm{\cimg(\thetavec) - \img}_2$.
Choosing the direction of maximum error, $\tilde{\tau} = \norm{\nabla b(\img)}_2 \max_{\thetavec} \norm{\cimg(\thetavec) - \img}_2^2 / \max_{\thetavec} \norm{\cimg(\thetavec) - \img}_2$.
The maximum quantization error is $n_p\Delta/2$, which yields our result.
\end{proof} 
In our experiments, we set $\tau = \alpha\tilde{\tau}$ and we vary the value of $\alpha$ to explore different trade-offs. We choose the Lagrangian as  
\begin{equation}
    \lambda = \tau c\, 2^{(\mathrm{QP} - 12)/3},
\end{equation}
following from the expectation of the regularized LNRM under the uniform noise model, $\expec\left( d(\img, \cimg)\right) = \tau \expec(\vert \vert \img - \cimg \vert \vert^2_2) = \tau n_p\Delta^2/12$. %This is the value of $\lambda$ we will consider in our experiments. 

\subsection{Transform domain evaluation}
Let $\bt U$ be an orthogonal transform \cite{strang_discrete_1999}, and define the gradient in the transform domain by $\bt t_i(\img) = \bt U^\top \nabla b_i(\img)$, for $i = 1, \hdots, n_b$. Then, we can re-write \eqref{eq:sse_reg} block-wise in transform domain as
 \begin{equation}
     d(\bt z_i, \hat{\bt z}_i(\theta_i)) = \bt t_i(\img)^\top (\hat{\bt z}_i(\theta_i) - \bt z_i) + \tau \, \norm{\hat{\bt z}_i(\theta_i) - \bt z_i}_2^2.
 \end{equation}
 with $\bt z_i = \bt U^\top \bt x_i$ and $\hat{\bt z}_i(\thetavec)$ being the compressed version of $\bt z_i$ with parameters $\thetavec$.  Therefore, LNRM-RDO can be written as 
 \begin{equation}
 \label{eq:trans_domain}
     \theta_i^\star = \argmin_{\theta_i\in\Theta_i}\, d(\bt z_i, \hat{\bt z}_i(\theta_i))+  \lambda\, r_i(\hat{\bt z}_i(\theta_i)), \quad i = 1, \hdots, n_b.
 \end{equation}
 Since the rate is measured in the transform domain, evaluating this expression is more efficient than computing \eqref{eq:sse_reg}.
\subsection{Complexity analysis}
\label{sec:complexity}
The gradient is computed \textit{only once} regardless of the number of RDO options. We use the Pytorch implementation of autodiff \cite{paszke_automatic_2017}. We need a forward and a backward pass of the NRM. Since the backward pass has roughly twice the complexity of the forward pass \cite{sepehri_hierarchical_2024}, computing the gradient is approximately three times more complex than evaluating the metric. After computing the transform-domain version of the gradient, evaluating the cost function has twice the complexity of computing the SSE. 
For NRMs based on neural networks, the FLOPs per pixel (FLOPs/px) with test images of size $224\times 224$ are 200 kFLOPs/px for ARNIQA \cite{hosu2024uhd}, and 82 kFLOPs/px for VSFA. Our method only changes encoding complexity; a standard decoder can be used with no changes in its complexity. 
\section{Empirical evaluation}
\label{sec:exper}
To guarantee that the compressed image converges to the input as bitrate increases, we use 4:4:4 AVC baseline\footnote{This avoids  potential denoising effects due to chroma subsampling. The same method can be applied to any other configurations such as 4:2:0.}. We compare our results with AVC using SSE-RDO. 
To compress the color channels, we add $3$ to the QP value of the luma channel. We run our experiments on an Intel(R) CPU E5-2667 with a NVIDIA Geforce RTX 3090 (24GB VRAM). For images and I frames, we use RDO to choose block-partitioning (between $4\times 4$ and $16\times 16$) and quantization step (we explore $\Delta \mathrm{QP} = -4, -3, \hdots, 3, 4$ at the block level). For P frames, we use our modified RDO only to choose block partitioning ($16\times 16$, $16\times 8$, $8\times 16$, $8\times 8$, $8\times 4$, $4\times 8$, $4\times 4$). For LNRM-RDO, we add SSE regularization \eqref{eq:sse_reg}. Let $\tilde{\tau}$ be the value in \eqref{eq:tau}; we set $2\tilde{\tau}$, $\tilde{\tau}$, and $\tilde{\tau}/2$. We set $\Delta(\mathrm{QP}) = 2^{(\mathrm{QP} - 4)/6}$ \cite{ma2005study}.

\subsection{Compression experiments}
For images, since our method can also be applied to pristine content, we consider both pristine and UGC inputs. In both cases, we will compare three versions of AVC: SSE-RDO (baseline), BRISQUE-RDO, and ARNIQA-RDO, with $\mathrm{QP}\in\lbrace 25, 28, 31, 34, 37\rbrace$. 

\begin{figure*}
    \centering
    \includegraphics[width=\linewidth]{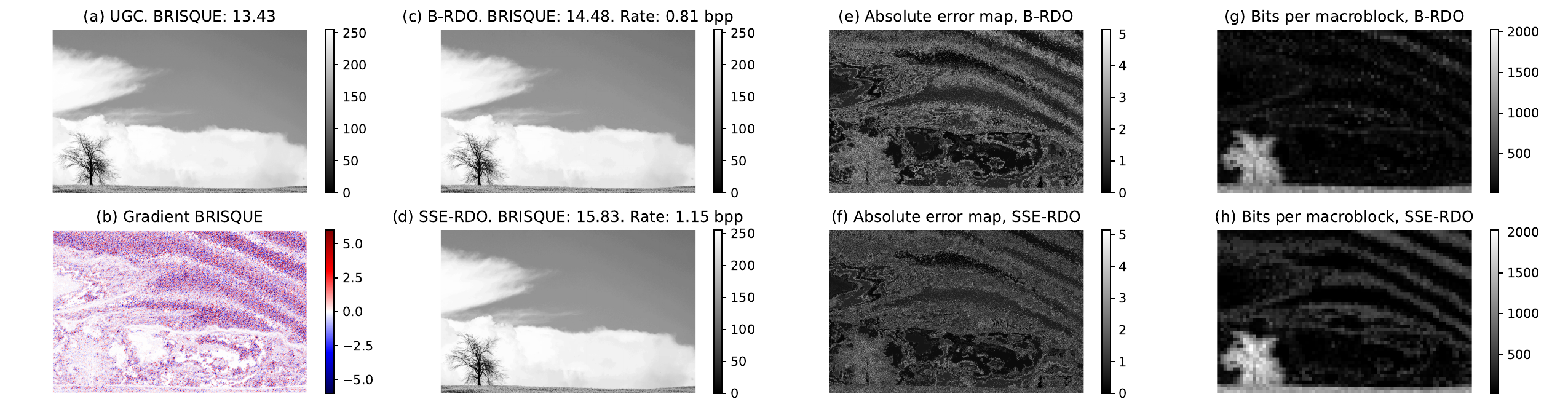}
    \caption{Synthetic UGC (a) and its BRISQUE gradient (b), showing that perturbations in noisy regions, e.g., the banding artifacts in the sky, contribute less to BRISQUE. Thus, we can punish these regions during compression. We compress using BRISQUE-RDO (B-RDO) (c) and SSE-RDO (d). As the error (e-f) and bitrate (g-h) maps reveal, BRISQUE-RDO avoids allocating excessive bitrate to the banding artifacts, reaching a better BRISQUE score than SSE-RDO with less bitrate.}
\label{fig:perceptual}
\vspace{-1em}
\end{figure*}

Regarding UGC images, we consider a set of $85$ frames of resolution $480\times 360$ pixels sampled from the YouTube UGC dataset \cite{wang2019youtube}. We show the BD-rate results in Table~\ref{tab:metrics_results} (top) for different values of $\tau$. We show the average rate-distortion curves using $\tau = \tilde{\tau}/2$ in Fig.~\ref{fig:rd_pristine} (d-f). We observe that using a larger $\tau$ improves the preservation of the original signal as measured by SSE, but comes at the cost of degrading the performance of the target metric. We provide a perceptual example using a synthetic UGC from the CLIC dataset \cite{CLIC2022} with banding artifacts in Fig.~\ref{fig:perceptual}. To encode this image, we considered SSE-RDO and BRISQUE-RDO with $\mathrm{QP}=20$. For BRISQUE-RDO, we consider regularization with $\tau = \tilde{\tau}/2$. Results show that BRISQUE-RDO can account for existing artifacts in the input UGC, reducing the overall bitrate needed to encode the image. 

For pristine content, we consider $24$ images of resolution $512\times768$ from the KODAK dataset \cite{kodak1993kodak}. Average bitrate savings are shown in Table~\ref{tab:metrics_results} (bottom) and we provide the rate-distortion curves for the average of the whole dataset in Fig.~\ref{fig:rd_pristine} (a-c). We observe that the BRISQUE/ARNIQA score of the input for pristine images is better than for UGC. Regarding the compressed image, convergence to the quality of the input image happens at higher bitrate. 

\begingroup
	\begin{table}[t]
		\centering
        \renewcommand{\arraystretch}{0.9}
        \setlength{\tabcolsep}{4pt} % Default value: 6pt
		\begin{tabular}{llcccc}
			\toprule & \textbf{Method}
			& \textbf{PSNR} [\%]  & \textbf{BRISQUE} [\%] &
   \textbf{ARNIQA} [\%] \\
			\midrule
   
            \multirow{8}{*}{\rotatebox{90}{\textbf{ UGC}}} &
			B-RDO, $ 2\tilde{\tau}$ & ${0.60}$ & $-27.29$ & $0.13$ \\	
			& A-RDO, $ 2\tilde{\tau}$ & $\mathbf{0.53}$ & $0.14$ & $-30.43$ \\		
			\cmidrule{2-5}  
            & B-RDO, $ \tilde{\tau}$ & $1.85$ & $-38.50$ & $0.37$ \\	
			&A-RDO, $ \tilde{\tau}$ & $1.83$ & $-0.89$ & $-41.37$ \\	
			\cmidrule{2-5}
			& B-RDO, $ \tilde{\tau}/2$ & $3.35$ & $\mathbf{-58.94}$ & $0.72$ \\	
			& A-RDO, $ \tilde{\tau}/2$ & $3.54$ & $0.69$ & $\mathbf{-57.26}$ \\
			\midrule
            \midrule
            \multirow{1}{*}{\rotatebox{90}{\textbf{KD \hspace{0.05em}}}}
			& B-RDO, $\tilde{\tau}/2$ & $3.54$ & $-40.21$ & $-0.54$ \\	
			&A-RDO, $ \tilde{\tau}/2$ & $3.72$ & $-1.40$ & $-55.20$ \\    		\bottomrule
    		\end{tabular}		
		\caption{BD-rate saving \cite{bjontegaard_calculation_2001} with respect to SSE-RDO, for images from the KODAK dataset (KD) and the YouTube UGC dataset \cite{wang2019youtube} using BRISQUE-RDO (B-RDO) and ARNIQA-RDO (A-RDO). More negative values are better, the best value is in boldface for a given metric. Each method optimizes the metric they target. More regularization improves SSE but worsens the target metric.}
		\label{tab:metrics_results}
        %\vspace{-1.2em}
    \end{table}	 
\endgroup
For video sequences, we use AVC in IPP... configuration. We set the GOP size to $30$ and let the quality parameter range in $\mathrm{QP}\in\lbrace 27, 30, 33, 36, 39\rbrace$. We use $85$ sequences of $480\times 360$ pixels from different videos in the YouTube dataset \cite{wang2019youtube}. Average bitrate savings are shown in Table~\ref{tab:video_ugc} for VSFA and BRISQUE, while the average RD curves for the same metrics are shown for $\tau = 2\tilde{\tau}$ in Fig.~\ref{fig:video_exper}. We observe the same trend as for images.

\begin{figure}[t]
    \centering
\includegraphics[width=\linewidth]{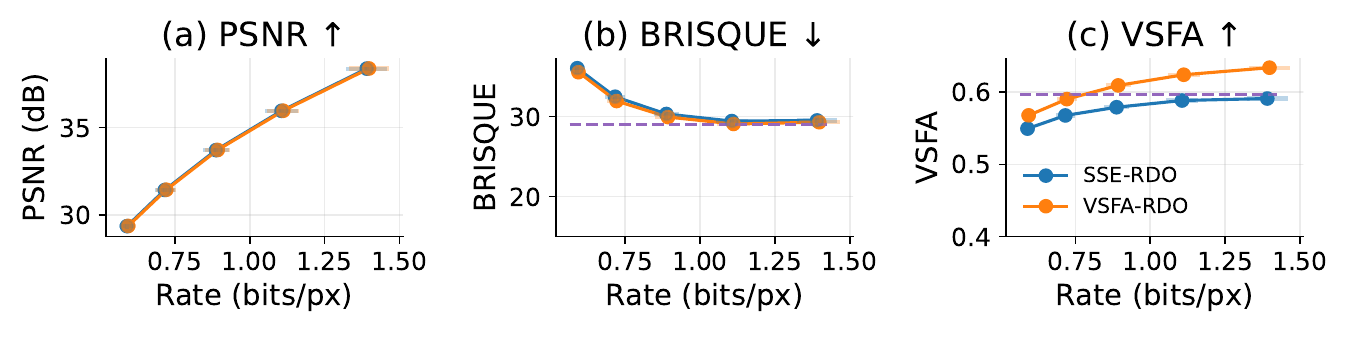}
    \caption{RD curves averaged across GOPs sampled from the YouTube-UGC dataset, using SSE-RDO and VSFA-RDO. We include the standard error for both axes. The average input quality is shown with a purple line. We optimize the target metric. VSFA-RDO can improve the VSFA score of the input UGC.}
    \label{fig:video_exper}
\end{figure}
\subsection{Runtime complexity}
We consider 1) the runtime to compute the gradient of a single frame for a given metric–for BRISQUE, ARNIQA, and VSFA–and 2) runtime increase for compressing the same content using LNRM-RDO and SSE-RDO. We report the average of $30$ images for image metrics and $30$ GOPs for VSFA, all with resolution $480\times 360$ pixels. We summarize our results in Table~\ref{tab:complexity_results}. ARNIQA requires more processing time than VSFA, as predicted in Sec.~\ref{sec:complexity}; although ARNIQA and VSFA entail more operations than BRISQUE, parallel processing in the GPU offsets the complexity. VSFA is compatible with real-time streaming at $60$ frames per second, despite being complex compared to other metrics \cite{wu2022fast}. Introducing the gradient adds a small computational overhead on the encoder runtime.

    \begin{table}[t]
        \centering
\renewcommand{\arraystretch}{0.9}		        
\setlength{\tabcolsep}{3pt} % Default value: 6pt
        \begin{tabular}{lccc}
            \toprule \textbf{Method}
            & \textbf{PSNR} [\%] & \textbf{VSFA} [\%] & \textbf{BRISQUE} [\%] \\
            \midrule
            VSFA-RDO, $ 2\tilde{\tau}$ & $\mathbf{0.47}$ & $-24.26$ & $\mathbf{-2.71}$ \\		
            VSFA-RDO, $ \tilde{\tau}$ &$1.62$ & $\mathbf{-34.53}$ & $-1.42$  \\	
            \bottomrule
        \end{tabular}
        \caption{BD-rate saving for $85$ video segments from the YouTube UGC dataset with respect to SSE-RDO. More negative is better. We show in boldface the best result for each metric.}
        \label{tab:video_ugc}
    \end{table}	     
    \begin{table}[ht!]
        \centering
\renewcommand{\arraystretch}{0.9}		        
\setlength{\tabcolsep}{5pt} % Default value: 6pt
        \begin{tabular}{lccc}
            \toprule \textbf{Metric}
            & \textbf{BRISQUE} & \textbf{ARNIQA} & \textbf{VSFA} \\
            \midrule
             Gradient computation &$0.034$ s & $0.025$ s & $0.015$ s  \\	
            Complexity overhead   & $2.12 \, \%$ & $2.09 \, \%$ & $3.51 \, \%$ \\	
            \bottomrule
        \end{tabular}
        \caption{Runtime to compute the gradient for a single frame  (top) and increase in encoder runtime  compared to SSE-RDO (bottom).}
\label{tab:complexity_results}
    \end{table}	     
\section{Conclusion}
In this paper, we tackled the challenge of compressing UGC videos by incorporating a non-reference quality metric (NRM) as the distortion term in rate-distortion optimization (RDO). Leveraging a Taylor expansion, we derived a distortion term that can be efficiently evaluated block-wise in the transform domain while capturing the perceptual information embedded in the NRM. To mitigate large deviations from the input and ensure stability, we introduced SSE regularization, along with explicit expressions for the regularization parameter and the Lagrange multiplier. Experimental results demonstrated substantial gains in the target NRM. Moreover, these improvements were achieved with minimal computational overhead. Moving forward, we aim to integrate the linearized cost function as the primary rate-distortion metric to train learned codecs \cite{balle_nonlinear_2020}. 

\bibliographystyle{IEEEbib}
\bibliography{IEEEabrv,conference_101719}

\begin{thebibliography}{10}

\bibitem{wang2019youtube}
Y. Wang, S. Inguva, and B. Adsumilli,
\newblock ``{YouTube} {UGC} dataset for video compression research,''
\newblock in {\em Proc. IEEE Intl. Work. on Mult. Signal Process.} IEEE, 2019, pp. 1--5.

\bibitem{seufert2014survey}
M. Seufert, S. Egger, M. Slanina, T. Zinner, et~al.,
\newblock ``A survey on quality of experience of {HTTP} adaptive streaming,''
\newblock {\em IEEE Comms. Surv. \& Tutor.}, vol. 17, no. 1, pp. 469--492, 2014.

\bibitem{sullivan_rate_1998}
G.~J. Sullivan and T. Wiegand,
\newblock ``Rate-distortion optimization for video compression,''
\newblock {\em IEEE Signal Process. Mag.}, vol. 15, no. 6, pp. 74--90, 1998.

\bibitem{pavez2022compression}
E. Pavez, E. Perez, X. Xiong, A. Ortega, and B. Adsumilli,
\newblock ``Compression of user generated content using denoised references,''
\newblock in {\em Proc. IEEE Int. Conf. Image Process.} IEEE, 2022, pp. 4188--4192.

\bibitem{wang_multiscale_2003}
Z. Wang, E.~P. Simoncelli, and A.~C. Bovik,
\newblock ``Multiscale structural similarity for image quality assessment,''
\newblock in {\em Proc. Asilomar Conf. on Signals, Sys. \& Comput.} IEEE, 2003, vol.~2, pp. 1398--1402.

\bibitem{ling2020towards}
S. Ling, Y. Baveye, P. Le~Callet, J. Skinner, and I. Katsavounidis,
\newblock ``Towards perceptually-optimized compression of user generated content ({UGC}): Prediction of {UGC} rate-distortion category,''
\newblock in {\em Proc. IEEE Intl. Conf. on Mult. and Expo}. IEEE, 2020, pp. 1--6.

\bibitem{john2020rate}
S. John, A. Gadde, and B. Adsumilli,
\newblock ``Rate distortion optimization over large scale video corpus with machine learning,''
\newblock {\em arXiv preprint arXiv:2008.12408}, 2020.

\bibitem{yu2021predicting}
X. Yu, N. Birkbeck, Y. Wang, C.~G. Bampis, et~al.,
\newblock ``Predicting the quality of compressed videos with pre-existing distortions,''
\newblock {\em IEEE Trans. Image Process.}, vol. 30, pp. 7511--7526, 2021.

\bibitem{xiong_rate_2023}
X. Xiong, E. Pavez, A. Ortega, and B. Adsumilli,
\newblock ``Rate-distortion optimization with alternative references for {UGC} video compression,''
\newblock in {\em Proc. IEEE Int. Conf. Acoust., Speech, and Signal Process.}, 2023, pp. 1--5.

\bibitem{li_asymptotic_1999}
J. Li, N. Chaddha, and R.~M. Gray,
\newblock ``Asymptotic performance of vector quantizers with a perceptual distortion measure,''
\newblock {\em IEEE Trans. Inform. Theory}, vol. 45, no. 4, pp. 1082--1091, 1999.

\bibitem{paszke_automatic_2017}
A. Paszke, S. Gross, S. Chintala, G. Chanan, et~al.,
\newblock ``Automatic differentiation in {Pytorch},''
\newblock 2017.

\bibitem{wiegand_lagrange_2001}
T. Wiegand and B. Girod,
\newblock ``Lagrange multiplier selection in hybrid video coder control,''
\newblock in {\em Proc. IEEE Int. Conf. Image Process.} 2001, vol.~2, pp. 542--545, IEEE.

\bibitem{wiegand_overview_2003}
T. Wiegand, G. Sullivan, G. Bjontegaard, and A. Luthra,
\newblock ``Overview of the {H}.264/{AVC} video coding standard,''
\newblock {\em IEEE Trans. Circuits Syst. Video Technol.}, vol. 13, no. 7, pp. 560--576, July 2003.

\bibitem{mittal2012no}
A. Mittal, A.~K. Moorthy, and A.~C. Bovik,
\newblock ``No-reference image quality assessment in the spatial domain,''
\newblock {\em IEEE Trans. Image Process.}, vol. 21, no. 12, pp. 4695--4708, 2012.

\bibitem{agnolucci2024arniqa}
L. Agnolucci, L. Galteri, M. Bertini, and A. Del~Bimbo,
\newblock ``Arniqa: Learning distortion manifold for image quality assessment,''
\newblock in {\em Proc. IEEE/CVF Wint. Conf. on Apps. of Comp. Vis.}, 2024, pp. 189--198.

\bibitem{li2019quality}
D. Li, T. Jiang, and M. Jiang,
\newblock ``Quality assessment of in-the-wild videos,''
\newblock in {\em Proc. ACM Intl. Conf. on Mult.}, 2019, pp. 2351--2359.

\bibitem{kodak1993kodak}
E. Kodak,
\newblock ``Kodak lossless true color image suite,''
\newblock {\em URL http://r0k. us/graphics/kodak}, 1993.

\bibitem{bross2021overview}
B. Bross, Y.-K. Wang, Y. Ye, S. Liu, et~al.,
\newblock ``Overview of the versatile video coding ({VVC}) standard and its applications,''
\newblock {\em IEEE Trans. Circuits Syst. Video Technol.}, vol. 31, no. 10, pp. 3736--3764, 2021.

\bibitem{han2021technical}
J. Han, B. Li, D. Mukherjee, C.-H. Chiang, et~al.,
\newblock ``A technical overview of {AV1},''
\newblock {\em Proc. of the IEEE}, vol. 109, no. 9, pp. 1435--1462, 2021.

\bibitem{everett_generalized_1963}
H. Everett~III,
\newblock ``Generalized {Lagrange} multiplier method for solving problems of optimum allocation of resources,''
\newblock {\em Operations research}, vol. 11, no. 3, pp. 399--417, 1963.

\bibitem{ortega_rate-distortion_1998}
A. Ortega and K. Ramchandran,
\newblock ``Rate-distortion methods for image and video compression,''
\newblock {\em IEEE Signal Process. Mag.}, vol. 15, no. 6, pp. 23--50, Nov. 1998.

\bibitem{fernandez2024feature}
S. Fernández-Menduiña, E. Pavez, and A. Ortega,
\newblock ``Feature-preserving rate-distortion optimization in image coding for machines,''
\newblock in {\em Proc. Intl. Works. on Mult. Signal Process.}, 2024, pp. 1--6.

\bibitem{ringis_disparity_2023}
D.~J. Ringis, Vibhoothi, F. Piti{\'e}, and A. Kokaram,
\newblock ``The disparity between optimal and practical {Lagrangian} multiplier estimation in video encoders,''
\newblock {\em Front. in Signal Process.}, vol. 3, pp. 1205104, 2023.

\bibitem{ponomarenko2015image}
N. Ponomarenko, L. Jin, O. Ieremeiev, V. Lukin, et~al.,
\newblock ``Image database {TID2013}: Peculiarities, results and perspectives,''
\newblock {\em Sig. proc.: Image comms.}, vol. 30, pp. 57--77, 2015.

\bibitem{linder_high-resolution_1999}
T. Linder, R. Zamir, and K. Zeger,
\newblock ``High-resolution source coding for non-difference distortion measures: multidimensional companding,''
\newblock {\em IEEE Trans. Inform. Theory}, vol. 45, no. 2, pp. 548--561, Mar. 1999.

\bibitem{liu2024defense}
Y. Liu, C. Yang, D. Li, J. Ding, and T. Jiang,
\newblock ``Defense against adversarial attacks on no-reference image quality models with gradient norm regularization,''
\newblock in {\em Proc. IEEE/CVF Conf. on Comp. Vis. and Patt. Recog.}, 2024, pp. 25554--25563.

\bibitem{berger2003rate}
T. Berger,
\newblock ``Rate-distortion theory,''
\newblock {\em Wiley Encyclopedia of Telecommunications}, 2003.

\bibitem{gish1968asymptotically}
H. Gish and J. Pierce,
\newblock ``Asymptotically efficient quantizing,''
\newblock {\em IEEE Trans. Inform. Theory}, vol. 14, no. 5, pp. 676--683, 1968.

\bibitem{strang_discrete_1999}
G. Strang,
\newblock ``The discrete cosine transform,''
\newblock {\em SIAM review}, vol. 41, no. 1, pp. 135--147, 1999.

\bibitem{sepehri_hierarchical_2024}
Y. Sepehri, P. Pad, A.~C. Yüzügüler, P. Frossard, and L.~A. Dunbar,
\newblock ``Hierarchical training of deep neural networks using early exiting,''
\newblock {\em IEEE Trans. on Neural Nets. and Learn. Sys.}, pp. 1--15, 2024.

\bibitem{hosu2024uhd}
V. Hosu, L. Agnolucci, O. Wiedemann, D. Iso, and D. Saupe,
\newblock ``Uhd-iqa benchmark database: Pushing the boundaries of blind photo quality assessment,''
\newblock {\em arXiv preprint arXiv:2406.17472}, 2024.

\bibitem{ma2005study}
S. Ma, W. Gao, D. Zhao, and Y. Lu,
\newblock ``A study on the quantization scheme in {H. 264/AVC} and its application to rate control,''
\newblock in {\em Proc. Pac. Rim Conf. on Mult., Part III 5}. Springer, 2004, pp. 192--199.

\bibitem{CLIC2022}
``5th {Challenge on Learned Image Compression dataset},'' Online, June 2022.

\bibitem{bjontegaard_calculation_2001}
G. Bjontegaard,
\newblock ``Calculation of average {PSNR} differences between {RD}-curves,''
\newblock {\em ITU SG16 Doc. VCEG-M33}, 2001.

\bibitem{wu2022fast}
H. Wu, C. Chen, J. Hou, L. Liao, et~al.,
\newblock ``Fast-{VQA}: Efficient end-to-end video quality assessment with fragment sampling,''
\newblock in {\em Euro. Conf. on Compt. Vis.} Springer, 2022, pp. 538--554.

\bibitem{balle_nonlinear_2020}
J. Ball{\'e}, P.~A. Chou, D. Minnen, S. Singh, et~al.,
\newblock ``Nonlinear transform coding,''
\newblock {\em IEEE Journal of Sel. Top. in Sig. Process.}, vol. 15, no. 2, pp. 339--353, 2020.

\end{thebibliography}

\end{document}